\documentclass[journal]{IEEEtran}

\usepackage{amsmath}
\usepackage{amsfonts}
\usepackage{amssymb}
\usepackage{amsthm}
\usepackage{tabularx}
\usepackage{mathrsfs} 

\usepackage{url}
\usepackage{hyperref}

\newtheorem{Proposition}{Proposition}

\newtheorem{Observation}{Observation}

\usepackage{stfloats}
\usepackage{float}
\usepackage{graphicx}
\usepackage{cite}
\usepackage{xcolor}
\usepackage{subfigure}

\makeatletter
\def\blfootnote{\xdef\@thefnmark{}\@footnotetext}
\makeatother

\begin{document}
	
		\title{\LARGE{Physical Layer Security over Fluid Reconfigurable Intelligent Surface-assisted Communication Systems}} 
	\author{Masoud~Kaveh,~\IEEEmembership{Member},~\textit{IEEE},~Farshad~Rostami~Ghadi,~\IEEEmembership{Member},~\textit{IEEE}, Francisco Hernando-Gallego, ~Diego~Mart\'in,
		Kai-Kit~Wong,~\IEEEmembership{Fellow},~\textit{IEEE}, and Riku Jäntti,~\IEEEmembership{Senior Member},~\textit{IEEE}
	}
	\maketitle
	\begin{abstract}
	This letter investigates the secrecy performance of wireless communication systems assisted by a fluid reconfigurable intelligent surface (FRIS). Unlike conventional reconfigurable intelligent surfaces (RISs) with fixed geometries, FRISs dynamically select a subset of reflective elements based on real-time channel conditions, offering enhanced spatial diversity and adaptability. Using this foundation, we model a secure downlink scenario where a base station communicates with a legitimate user in the presence of an eavesdropper, and the propagation is assisted by a FRIS with a limited number of elements set to the ON state. We analyze the system’s secrecy performance under spatial correlation by deriving analytical lower and upper bounds for the secrecy outage probability (SOP) and average secrecy capacity (ASC), respectively. Our results demonstrate that FRIS effectively enables secure communication under spatial correlation. Even with partial activation, FRIS significantly outperforms conventional RISs in enhancing secrecy performance under varying deployment densities and element correlations.
	\end{abstract}
	\maketitle
\blfootnote{The work of M. Kaveh and R. Jäntti has received funding from the SNS JU under the EU’s Horizon Europe Research and Innovation Programme under Grant Agreement No. 101192113 (Ambient-6G).
The work of F. Rostami Ghadi was supported by the EU’s Horizon 2022 Research and Innovation Programme through Marie Skłodowska-Curie under Grant 101107993. The work of K. K. Wong is supported by the Engineering and Physical Sciences Research Council (EPSRC) under Grant EP/W026813/1.}
\blfootnote{\noindent M. Kaveh and R. Jäntti are with the Department of Information and Communication Engineering, Aalto University, Espoo, Finland. (e-mail: $\rm masoud.kaveh@aalto.fi, riku.jantti@aalto.fi$).}
\blfootnote{\noindent F. Rostami Ghadi is with the Department of Signal Theory, Networking and Communications, 
University of Granada, Granada, Spain. (e-mail: $\rm f.rostami@ugr.es$).}
\blfootnote{\noindent K. K. Wong is with the Department of Electronic and Electrical
Engineering, University College London, London, U.K., the Yonsei Frontier Lab, Yonsei University, Seoul, South Korea. (e-mail: $\rm kai\text{-}kit.wong@ucl.ac.uk$)}
\blfootnote{\noindent F. H. Gallego and D. Mart\'in are with the Department of Applied Mathematics and Computer Science, respectively, 
Escuela de Ingenier\'ia Inform\'atica de Segovia, 
Universidad de Valladolid, Segovia, Spain (e-mail:$\rm \{fhernando, diego.martin.andres\}@uva.es$)}
	
	\section{Introduction}\label{sec-intro}
\IEEEPARstart{R}{econfigurable} intelligent surfaces (RISs) have emerged as a promising potential to enhance spectral and energy efficiency in next-generation wireless networks \cite{liu2021ris}. By smartly altering the phase shifts of passive reflecting elements, RISs can reshape the propagation environment in a cost-effective manner \cite{Bilotti2024:SmartEM}. However, in practice, RIS performance faces several obstacles: a multiplicative path-loss from the cascaded double-fading link; channel estimation overhead that grows rapidly with the number of reflecting elements; and geometric rigidity that limits the exploitation of spatial diversity when physical space is abundant or irregularly shaped, ultimately constraining performance \cite{elmossallamy2020ris}.

To overcome the spatial rigidity of conventional RISs, recent research has explored integrating the concept of fluid antenna systems (FAS) \cite{Wong2021:FAS,Fluid-survey} into the metasurface, whereas the former was originally proposed to reconfigure the shape and position of antennas on the user side \cite{Lu2025:FluidAntennas}, and has shown great potential to enhance spatial adaptability in various wireless applications \cite{Ghadi2024:FAS_BC}. Building on this idea, the synergy between FAS and RIS was first examined in joint system models \cite{Ghadi2024:RIS_FAS, Ghadi2025:RIS_FAS_PLS}. 
This progression has then naturally inspired \emph{porting the fluid principle onto the metasurface itself}, giving rise to the concept of fluid RIS (FRIS) \cite{Salem2025:FRIS_FirstLook}, where each reflecting unit behaves as a fluid subelement whose location can be reconfigured within a predefined aperture \cite{Xiao2025:FRIS_Joint}, thus exploiting spatial diversity and improving the robustness of the link. Additionally, the FRIS concept has recently evolved into the fluid integrated reflecting and emitting surface (FIRES), enabling 360$^\circ$ coverage by supporting both reflection and transmission \cite{Ghadi2025:FIRES}. Furthermore, \cite{Ghadi2025:FRIS_Arxiv} developed analytical frameworks for ergodic capacity and outage probability in FRIS-assisted links, showing substantial performance gains over conventional RISs. 

Despite recent advances in FRIS, its potential to enhance the physical layer security (PLS) of wireless communication systems remains unaddressed; specifically, the ability of FRIS to secure transmissions against eavesdropping has not been studied yet.
Motivated by this gap, this paper presents an early analysis on the secrecy performance of a downlink wireless system where a base station communicates with a legitimate user via an FRIS in the presence of a passive eavesdropper.
In particular, based on the optimized FRIS framework introduced in \cite{Xiao2025:FRIS_Joint,Ghadi2025:FRIS_Arxiv}, we employ moment matching methods to derive closed-form expressions for the probability density function (PDF) and the cumulative distribution function (CDF) of the equivalent channel gains observed both in legitimate receiver and eavesdropper.
Subsequently, we derive analytical lower and upper bounds for the secrecy outage probability (SOP) and average secrecy capacity (ASC), respectively, and compare their behavior with that observed in conventional RIS-assisted systems. Our numerical results confirm that FRIS, even with partial activation ports, significantly enhances PLS performance compared to conventional RISs under various deployment configurations.

	\section{System Model}\label{sec-sys}
    
    \begin{figure}[t]
    \centering    \includegraphics[width=0.36\textwidth]{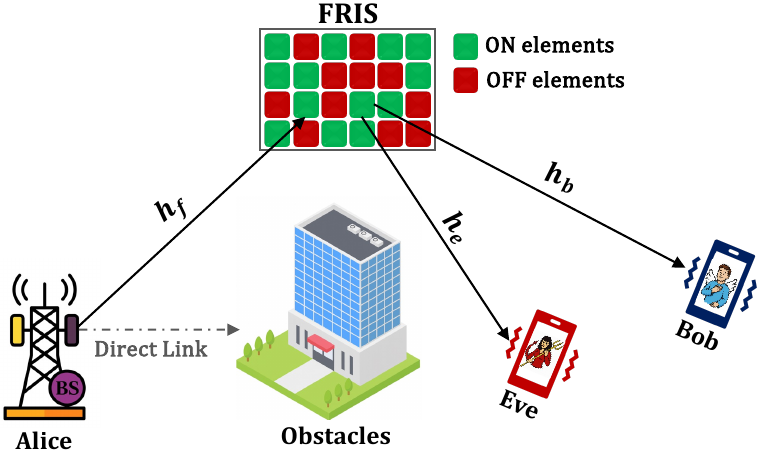}
    \caption{A secure FRIS-aided communication system with ON/OFF elements.}
    \label{fig_model}
\end{figure}
	We consider a secure wireless communication setup, as shown in Fig.~\ref{fig_model}, where a legitimate transmitter (Alice) aims to send confidential messages to a legitimate receiver (Bob) with the help of FRIS, while an eavesdropper (Eve) attempts to decode the desired signal. We assume that the direct links between Alice and nodes $u\in\left\{b,e\right\}$, which respectively represents Bob and Eve, are blocked due to obstacles. It is assumed that all nodes are equipped with fixed-position antennas, whereas the FRIS includes $M=M_x\times M_z$ reflecting elements, operating in one of two discrete modes ON or OFF. In the ON mode, the element responds to the incident electromagnetic wave, modifying its phase to contribute to the system's operation. When set to OFF, however, the element is connected to a matched load, effectively disconnecting it from the wave and ensuring no impact on the signal. The FRIS elements are arranged uniformly across a surface with dimensions defined as $W = W_x\lambda \times W_z\lambda$, where $\lambda$ represents the carrier wavelength. Each row contains $M_x$ elements spanning a total length of $W_x\lambda$, and each column consists of $M_z$ elements extending over $W_z\lambda$. Due to the close proximity of adjacent elements, spatial correlation effects cannot be ignored. To model this, we employ the Jakes' correlation model under rich scattering conditions. Hence, the spatial correlation coefficient between any two elements $i$ and $j$ is given by
\begin{align}
	j_{i,l} = \mathcal{J}_0\left(\frac{2\pi d_{i,l}}{\lambda}\right),
\end{align}
	where $\mathcal{J}_0(\cdot)$ denotes the zeroth-order Bessel function of the first kind, and $d_{i,l}$ represents the Euclidean distance between the two elements. The calculation of $d_{i,l}$ depends on the indexing convention used to map the one-dimensional (1D) index of an element to its two-dimensional (2D) coordinates. Therefore, $d^2_{i,l}$ can be defined as 
		\begin{multline}
		d^2_{i,l} = d_x^2\left(\mathrm{mod}\left(i,M_x\right)-\mathrm{mod}\left(l,M_x\right)\right)^2\\
		+d_z^2\left(\left\lfloor\frac{i}{M_x}\right\rfloor-\left\lfloor\frac{l}{M_x}\right\rfloor \right),
	\end{multline}
		where $d_x=\frac{W_x\lambda}{M_x}$ and $d_z=\frac{W_z\lambda}{M_z}$ are the physical inter-element spacing in each FRIS row and column, respectively. Consequently, the spatial correlation matrix $\mathbf{J}$ is given by
	\begin{align}
		\mathbf{J} = \begin{bmatrix}
			j_{1,1} & j_{1,2} & \cdots & j_{1,M} \\
			j_{2,1} & j_{2,2} & \cdots & j_{2,M} \\
			\vdots & \vdots & \ddots & \vdots \\
			j_{M,1} & j_{M,2} & \cdots & j_{M,M} \\
		\end{bmatrix}.
	\end{align}

Furthermore, we assume Alice has access to the instantaneous channel state information (CSI) of Bob but has no knowledge of the Eve’s CSI. Specifically, the FRIS element selection and phase shift configuration are designed solely based on Bob’s CSI to enhance the quality of the legitimate link. Consequently, the resulting FRIS configuration is optimized to maximize the signal power at Bob and is fixed during the transmission. Since Eve's location and channel conditions are unknown to the system, she observes the same FRIS configuration as Bob but from an unintended spatial direction. As a result, the equivalent channel observed by Eve is not beamformed or aligned to her channel, and the corresponding channel response appears statistically random from her perspective. 
	
Assuming that $M_\mathrm{ON}$ elements will be turned ON to modulate the incident signal, the received signal at the node $u\in\left\{b,e\right\}$ can be mathematically expressed as follows:
\begin{align}
	y_\mathrm{u} = \sqrt{PL_\mathrm{f} L_\mathrm{u}}\underset{H_{\mathrm{eq},u}}{\underbrace{\mathbf{h}_{u}^H\mathbf{J}^{\frac{1}{2}}\mathbf{S}^T_{M_\mathrm{ON}}\mathbf{\Phi}\mathbf{S}_{M_\mathrm{ON}}\mathbf{J}^{\frac{1}{2}}\mathbf{h}_\mathrm{f}}}x+z_\mathrm{u},\label{eq-y}
\end{align}
where $L_\mathrm{k} = \rho d_\mathrm{k}^{-\alpha}$ for $k \in \{f, u\}$ characterizes the large-scale path loss, where $\rho$ is the reference signal gain at a distance of one meter, $\alpha$ denotes the path-loss exponent, and $d_\mathrm{k}$ corresponds to the propagation distances across the Alice-to-FRIS, FRIS-to-Bob, and FRIS-to-Eve links. The transmit power from Alice is indicated by $P$, while $x$ is the data signal being transmitted. At the receiving node $u$, the noise $z_\mathrm{u}$ is modeled as circularly symmetric complex Gaussian noise with zero mean and variance $\sigma^2$, i.e., $z_\mathrm{u} \sim \mathcal{CN}(0, \sigma^2)$. The equivalent fading channel $H_\mathrm{eq}$ includes the channel vector $\mathbf{h}_\mathrm{k} \in \mathbb{C}^{M \times 1}$, whose elements are i.i.d. complex Gaussian random variables with zero mean and unit variance, i.e., $\mathbf{h}_\mathrm{k} \sim \mathcal{CN}(0, \mathbf{I}_M)$, modeling small-scale fading over all links. The matrix $\mathbf{\Phi} = \mathrm{diag}([\mathrm{e}^{j\phi_1}, \ldots, \mathrm{e}^{j\phi_M}]) \in \mathbb{C}^{M \times M}$ specifies the configurable phase shifts of the FRIS elements. The element selection is defined by $\mathbf{S}_{M_\mathrm{ON}}^T = [\mathbf{s}_1, \ldots, \mathbf{s}_{M_\mathrm{ON}}] \in \mathbb{R}^{M \times M_\mathrm{ON}}$, where each $\mathbf{s}_{m_\mathrm{ON}}$ for $m_\mathrm{ON} \in \mathcal{M}_\mathrm{ON} = \{1, \ldots, M_\mathrm{ON}\}$ denotes a selected column from the $M \times M$ identity matrix.

\section{Secrecy Performance Analysis}
This section first derives the PDF and CDF of the equivalent channel gains at Bob and Eve, then uses them to obtain closed-form results for SOP and ASC.

\subsection{Statistical Characterization}
From
\eqref{eq-y}, the received SNR at node $u$ is given by
\begin{align}
	\gamma_\mathrm{u}=\bar{\gamma}_\mathrm{u}L_\mathrm{f}L_\mathrm{u}\left|H_{\mathrm{eq},u}\right|^2,\label{eq-snr}
\end{align}
where $\bar{\gamma}_\mathrm{u} \triangleq P/\sigma^2_\mathrm{u}$ denotes the average transmit SNR, and $G_\mathrm{u}$ represents the equivalent channel power gain for the corresponding links, defined as
\begin{align}
G_\mathrm{u}\triangleq\left|H_{u,\mathrm{eq}}\right|^2= \left|\mathbf{h}_{u}^H\mathbf{J}^{\frac{1}{2}}\mathbf{S}^T_{M_\mathrm{ON}}\mathbf{\Phi}\mathbf{S}_{M_\mathrm{ON}}\mathbf{J}^{\frac{1}{2}}\mathbf{h}_\mathrm{f}\right|^2.
\end{align}
 
\begin{Proposition}\label{pro-1}
	The distribution of $G_\mathrm{b}$ is accurately approximated by the Gamma distribution, i.e., $G_\mathrm{b}\sim\mathrm{Gamma}\left(k_\mathrm{b},\theta_\mathrm{b}\right)$, with the following CDF and PDF:
	\begin{align}
		F_{G_\mathrm{b}}\left(g\right) = \frac{1}{\Gamma\left(k_\mathrm{b}\right)}\Upsilon\left(k_\mathrm{b},\frac{g}{\theta_\mathrm{b}}\right),\label{eq-cdf}
	\end{align} 
	\begin{align}
		f_{G_\mathrm{b}}\left(g\right) = \frac{1}{\theta_\mathrm{b}^{k_\mathrm{b}}\Gamma\left(k_\mathrm{b}\right)}g^{k_\mathrm{b}-1}\mathrm{e}^{-g/\theta_\mathrm{b}},\label{eq-pdf}
	\end{align}
	where $\Upsilon(a,b)$ is the lower incomplete gamma function, $k_\mathrm{b} =\frac{\left(\mathrm{tr}\left(\widetilde{\mathbf{J}}^2\right)\right)^2}{\mathrm{tr}\left(\widetilde{\mathbf{J}}^4\right)}$, $\theta_\mathrm{b} = \frac{\mathrm{tr}\left(\widetilde{\mathbf{J}}^4\right)}{\mathrm{tr}\left(\widetilde{\mathbf{J}}^2\right)}$, and $\widetilde{\mathbf{J}}=\mathbf{S}_{M_\mathrm{ON}} \mathbf{J} \mathbf{S}_{M_\mathrm{ON}}^T$.
\end{Proposition}
\begin{proof}The proof details are in Appendix \ref{app-1}. 
	\end{proof}
\begin{Proposition}\label{pro-2}
	The distribution of $G_\mathrm{e}$ is accurately approximated by the Exponential distribution, i.e., $G_\mathrm{e}\sim\mathrm{Exp}\left(\theta_\mathrm{e}\right)$, with the following CDF and PDF:
	\begin{align}
	F_{G_\mathrm{e}}\left(g\right) = 1-\mathrm{exp}\left(-\theta_\mathrm{e}g\right),
	\end{align}
		\begin{align} \label{eq-pdf-e}
		f_{G_\mathrm{e}}\left(g\right) = \theta_\mathrm{e}\mathrm{exp}\left(-\theta_\mathrm{e}g\right), 
	\end{align}
where $\theta_\mathrm{e}=\frac{1}{\mathrm{tr}\left(\widetilde{\mathbf{J}}^2\right)}$ and $\widetilde{\mathbf{J}}=\mathbf{S}_{M_\mathrm{ON}} \mathbf{J} \mathbf{S}_{M_\mathrm{ON}}^T$.
	\end{Proposition}
\begin{proof}The proof details are in Appendix \ref{app-2}. 
	\end{proof}

\subsection{ASC Analysis}
Since the received SNRs at Bob and Eve are random variables, the instantaneous secrecy capacity also becomes a random quantity. Therefore, the average secrecy capacity (ASC), is defined as the expected value of the instantaneous secrecy rate, i.e., 
\begin{align} \label{sec-def}
    C_\mathrm{s}(\gamma_\mathrm{b}, \gamma_\mathrm{e}) = \left[ \log_2(1+\gamma_\mathrm{b}) - \log_2(1+\gamma_\mathrm{e}) \right]^+,
\end{align}
in which $[\cdot]^+ \triangleq \max(0, \cdot)$ ensures that the secrecy capacity is non-negative. 
Using the PDF distributions derived in \eqref{eq-pdf} and \eqref{eq-pdf-e}, the ASC for the considered FRIS-assisted system is expressed as
\begin{align}
    \bar{C}_\mathrm{s} \triangleq \int_0^\infty \int_0^\infty C_\mathrm{s}(\gamma_\mathrm{b}, \gamma_\mathrm{e}) f_{\gamma_\mathrm{b}}(\gamma_\mathrm{b}) f_{\gamma_\mathrm{e}}(\gamma_\mathrm{e}) \,\mathrm{d}\gamma_\mathrm{b} \mathrm{d}\gamma_\mathrm{e}.
\end{align}
However, solving the above-mentioned integral is mathematically intractable. Hence, we adopt a bounding approach to obtain analytical insights.
\begin{Observation}\label{Pro_ASC}
The upper bound of the ASC for the considered FRIS-assisted system can be approximated as
\begin{align} \label{ASC pro}
    \bar{C}_\mathrm{s}^{\mathrm{U}} = \log_2\left( \frac{1 + \bar{\gamma}_\mathrm{b} L_\mathrm{f} L_\mathrm{b} k_\mathrm{b} \theta_\mathrm{b}}{1 + \bar{\gamma}_\mathrm{e} L_\mathrm{f} L_\mathrm{e}/{\theta_\mathrm{e}}}\right).
\end{align}
\end{Observation}
\begin{proof}
The proof details are in Appendix \ref{app_ASC}. 
\end{proof}

\subsection{SOP Analysis}
Given that the FRIS configuration is optimized based solely on Bob’s CSI and remains fixed during transmission, the random nature of the channel observed by Eve makes the SOP a meaningful and tractable metric, which quantifies the likelihood that the instantaneous secrecy capacity falls below a predefined target rate $R_\text{s}$, that is, $P_{\mathrm{sop}}=\Pr\left(C_\text{s}\le R_\text{s}\right)$. 
Now, by inserting \eqref{sec-def} into the SOP definition we have
\begin{align} 
P_\mathrm{sop}&=\Pr\left(\log_2\left(\frac{1+\gamma_\mathrm{b}}{1+\gamma_\mathrm{e}}\right)\le R_\mathrm{s}\right) 
\end{align}
In this regard, we present a closed-form approximation of the SOP lower bound by leveraging the statistical distributions of the equivalent channel gains at both Bob and Eve.
\begin{Observation} \label{Pro_SOP}
A closed-form lower bound for the SOP of the considered FRIS-assisted system is given by
\begin{equation} \label{Pro-SOP}
P_\mathrm{SOP}^{\mathrm{L}} = \frac{ \bar{\gamma}_\mathrm{b} \, L_\mathrm{b} \, \theta_\mathrm{b} \, \theta_\mathrm{e} }{ \bar{\gamma}_\mathrm{e} \, L_\mathrm{e} \, 2^{R_\mathrm{s}} \, \Gamma(k_\mathrm{b}) } \, G_{2,2}^{2,1} \left( \frac{ \bar{\gamma}_\mathrm{b} \, L_\mathrm{b} \, \theta_\mathrm{b} \, \theta_\mathrm{e} }{ \bar{\gamma}_\mathrm{e} \, L_\mathrm{e} \, 2^{R_\mathrm{s}} } \;\middle|\; \begin{matrix} -k_\mathrm{b},\, 0 \\ 0,\, -1 \end{matrix} \right),
\end{equation}
where $G_{p,q}^{m,n}(\cdot)$ and $\Gamma(\cdot)$ denote the Meijer G-function and gamma function, respectively.
\end{Observation}
\begin{proof}
The proof details are in Appendix \ref{app_SOP}. 
\end{proof}

\section{Numerical Results}\label{sec-num}

To validate the secrecy performance of the FRIS-aided communication system, Monte Carlo simulations are conducted with \(10^6 \) independent realizations. The carrier frequency is set to \( f_c = 2.4\,\mathrm{GHz} \), corresponding to a wavelength \( \lambda = c/f_c \). The FRIS is configured as a square array comprising \( M_\mathrm{x} = M_\mathrm{z} = 20 \) elements per row and column, resulting in a total of \( M = 400 \) elements. However, only \( M_\mathrm{ON} \leq M \) elements are adaptively selected. The inter-element spacing for FRIS is determined by a physical aperture of \( 3\lambda \) in both horizontal and vertical dimensions. For comparison, a conventional RIS with an equal number of elements (\( M_\mathrm{conv}=M_\mathrm{ON} \)) is deployed with a uniform square geometry and a fixed interelement distance of \( \lambda/2 \). 
The distances BS-to-FRIS, FRIS-to-Bob, and FRIS-to-Eve are set to 20\,m, 30\,m, and 30\,m, respectively, with a path-loss exponent \( \alpha = 2.5 \) and reference gain \( \rho = 1 \).
The secrecy rate threshold is set to \( R_\mathrm{s} = 1\,\mathrm{bps/Hz} \). The transmission power is set to 30\,dBm, while the noise power in Bob and Eve is $\sigma^2_\mathrm{b}=-90$ dBm and $\sigma^2_\mathrm{e}=-80$ dBm, respectively. All channels are modeled as spatially correlated Rayleigh fading.

\begin{figure*}[t]
\centering
    \begin{minipage}[b]{0.30\linewidth}
        \centering
        \raggedleft
\includegraphics[width=\linewidth]{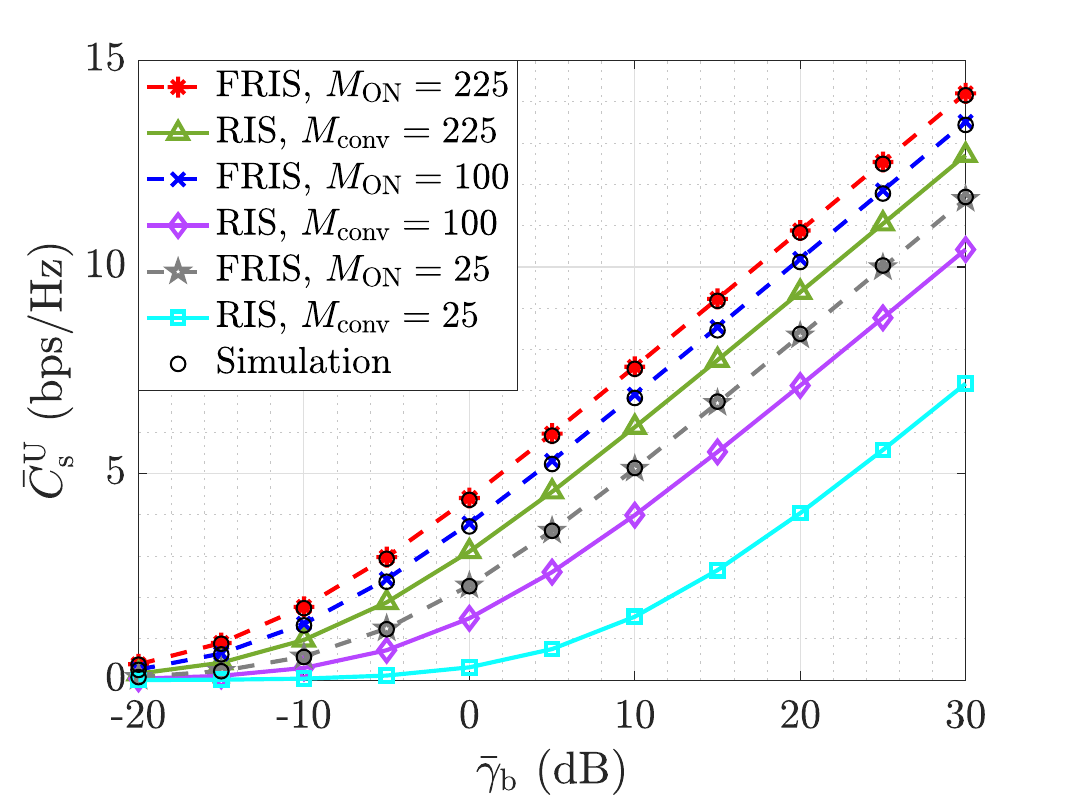}
\vspace{-15pt}
    \caption{ASC versus Bob's average SNR for different $M_\text{ON}$ and $M_\text{conv}$.}
    \label{fig-ASC-M}
    \end{minipage}
    \hfill
    \begin{minipage}[b]{0.30\linewidth}
        \centering
\includegraphics[width=\linewidth]{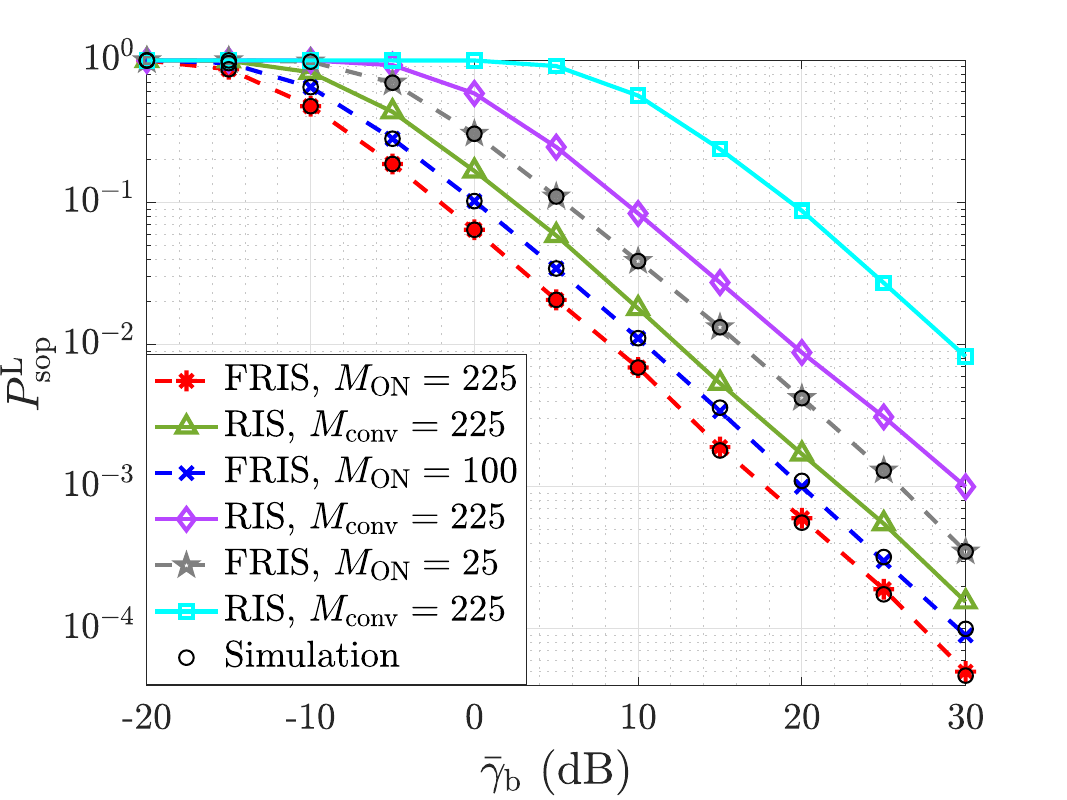}
\vspace{-15pt}
    \caption{SOP versus Bob's average SNR for different $M_\text{ON}$ and $M_\text{conv}$.}
    \label{fig-SOP-M}
   \end{minipage}
    \hfill
    \begin{minipage}[b]{0.30\linewidth}
       \raggedright
\includegraphics[width=\linewidth]{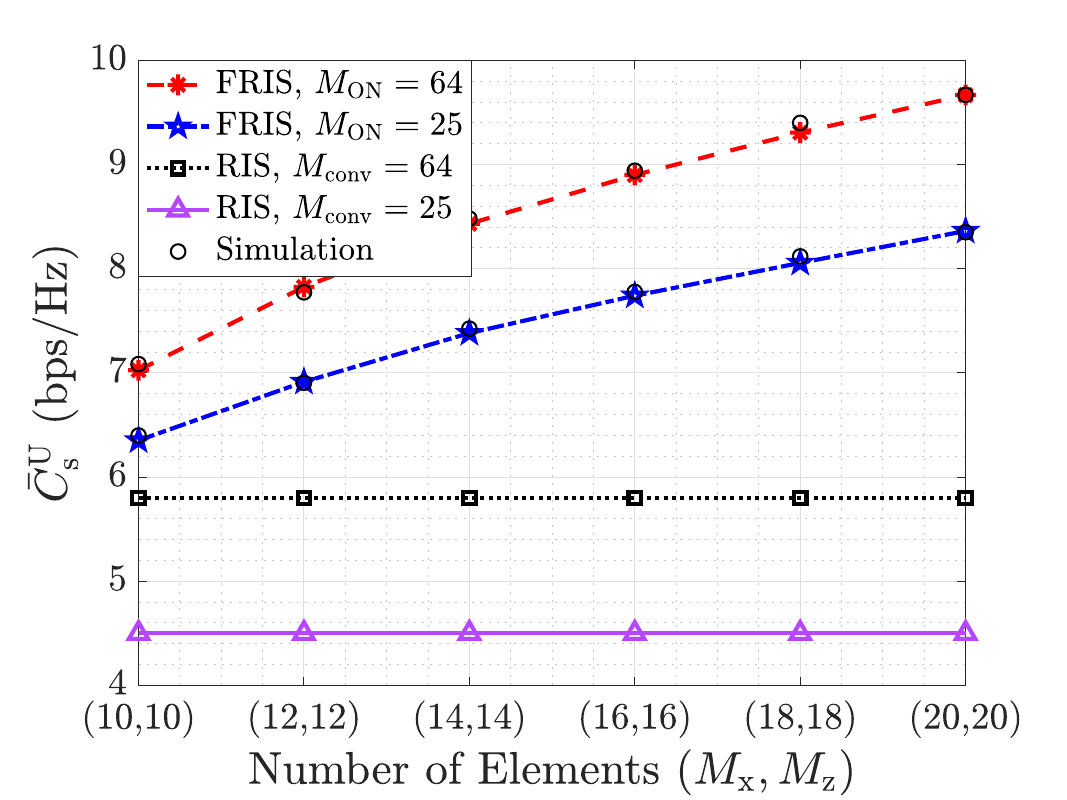}
\vspace{-15pt}
    \caption{ASC versus the number of elements $M = (M_\text{x}, M_\text{z} )$ for different $M_\text{ON}$.}
    \label{fig-ASC-Mx-Mz}
    \end{minipage}
\end{figure*}

Fig.~\ref{fig-ASC-M} shows the ASC versus Bob's average SNR $\bar{\gamma}_\mathrm{b}$ for different FRIS and RIS configurations.
The monotonic increase in the upper bound ASC with the average SNR at Bob is due to the Jensen-type bound, which increases in the mean of $\gamma_\mathrm{b}$. Under the proposed moment-matching, the equivalent channel gain at Bob is modeled as $G_\mathrm{b}\sim\mathrm{Gamma}(k_\mathrm{b},\theta_\mathrm{b})$ with
$k_\mathrm{b}$ and 
$\theta_\mathrm{b}$ shown in Proposition \ref{pro-1}, and
$\mathbb{E}[G_\mathrm{b}]=k_\mathrm{b}\theta_\mathrm{b}=\operatorname{tr}(\widetilde{\mathbf{J}}^{2})$.
For an FRIS, activating the best $M_\mathrm{ON}$ elements out of a larger pool increases $\operatorname{tr}(\widetilde{\mathbf{J}}^{2})$ while typically reducing the ratio $\operatorname{tr}(\widetilde{\mathbf{J}}^{4})/\big(\operatorname{tr}(\widetilde{\mathbf{J}}^{2})\big)^{2}$; this simultaneously raises the mean SNR $\mathbb{E}[\gamma_\mathrm{b}]$ and the shape factor $k_\mathrm{b}$, thus boosting $\bar{C}_s^{\mathrm{U}}$ and reducing dispersion. Conventional RIS with $M_\mathrm{conv}$ lacks this order-statistics (selection) gain and suffers more from spatial correlation, hence its lower curves. At low SNR, all schemes break down due to noise dominance while at moderate/high SNR the gap between FRIS and RIS widens. 

The impact of $\bar{\gamma}_\mathrm{b}$ on the SOP performance for different FRIS and RIS configurations is illustrated in Fig.~\ref{fig-SOP-M}.
Using the lower-bound SOP (representation in \eqref{eq:SOP_L_def})
with $G_\mathrm{b}\!\sim\!\mathrm{Gamma}(k_\mathrm{b},\theta_\mathrm{b})$ and $G_\mathrm{e}\!\sim\!\mathrm{Exp}(\theta_\mathrm{e})$, the high-$\bar{\gamma}_\mathrm{b}$ behavior is dominated by the left tail of $F_{\gamma_\mathrm{b}}(\cdot)$, yielding $
P_\mathrm{sop}^{\mathrm{L}}\propto \bar{\gamma}_\mathrm{b}^{-k_\mathrm{b}}$, for fixed $R_s$ and Eve statistics.
Thus larger $k_\mathrm{b}$ produces a steeper decay. FRIS increases $k_\mathrm{b}$ via element selection (lower $\operatorname{tr}(\widetilde{\mathbf{J}}^{4})$ relative to $\big(\operatorname{tr}(\widetilde{\mathbf{J}}^{2})\big)^2$), so its curves exhibit both a better coding gain (leftward shift from larger $\mathbb{E}[\gamma_\mathrm{b}]$) and a higher diversity order (steeper slope). A conventional RIS with fixed $M_\mathrm{conv}$ offers no selection freedom; its stronger correlation and smaller $k_\mathrm{b}$ lead to noticeably flatter decay. The close match between markers and the analysis validates the Gamma/Exponential approximation for Bob/Eve channels.

Fig.~\ref{fig-ASC-Mx-Mz} shows the ASC versus the number of total metasurface elements $M=(M_x,M_z)$ for fixed $M_\mathrm{ON}$.
With $M_\mathrm{ON}$ fixed, FRIS still exhibits increasing $\bar{C}_s^{\mathrm{U}}$ as the total surface element $M$ grows because the controller selects the best $M_\mathrm{ON}$ out of a larger candidate set; an order-statistics gain. As $M$ increases, the selected subset tends to have larger $\operatorname{tr}(\widetilde{\mathbf{J}}^{2})$ and smaller $\operatorname{tr}(\widetilde{\mathbf{J}}^{4})/\big(\operatorname{tr}(\widetilde{\mathbf{J}}^{2})\big)^{2}$, which raises both $\mathbb{E}[\gamma_\mathrm{b}]$ and $k_\mathrm{b}$. Although denser packing can elevate correlation locally, FRIS mitigates this by preferentially choosing stronger and more spatially separated elements (implicit de-correlation through selection). In contrast, the conventional RIS curves remain essentially flat because $M_\mathrm{conv}$ is fixed: without selection the effective statistics of $G_\mathrm{b}$ do not improve with $M$. The higher trajectory for larger $M_\mathrm{ON}$ reflects the additional array/selection gain at Bob and is consistent with the trends observed in Figs.~\ref{fig-ASC-M}–\ref{fig-ASC-Mx-Mz}.

\section{Conclusion}\label{sec-con}
This letter investigated the PLS performance of FRIS-assisted downlink transmission. By employing moment matching, we characterized the equivalent gains (Gamma for Bob and exponential for Eve), from which we derived a closed-form lower bound on the SOP and a tight Jensen-type upper bound on the ASC. Monte Carlo simulations validated the analysis and demonstrated that, even with partial activation, FRIS consistently outperforms conventional RIS in PLS performance with the same number of active elements. Moreover, the performance gains increase with $M_\mathrm{ON}$ and with aperture size (at fixed $M_\mathrm{ON}$), owing to selection/order-statistics and an effective increase in the diversity parameter $k_\mathrm{b}$. 

\appendices
\section{Proof of Proposition \ref{pro-1}}\label{app-1}
Since we assumed that the FRIS is configured based on the instantaneous CSI of Bob, and remains fixed throughout the transmission, the element selection matrix $\mathbf{S}_{M_\mathrm{ON}}$ and phase shift matrix $\boldsymbol{\Phi}$ are optimized to enhance the received signal power at Bob. Under this deterministic and optimized configuration, the equivalent channel gain at Bob is given by $G_\mathrm{b} = \left|\mathbf{h}_\mathrm{b}^H \mathbf{A} \mathbf{h}_\mathrm{f}\right|^2$, where $\mathbf{A} \triangleq \mathbf{J}^{1/2} \mathbf{S}_{M_\mathrm{ON}}^T \boldsymbol{\Phi} \mathbf{S}_{M_\mathrm{ON}} \mathbf{J}^{1/2}$ denotes the effective FRIS reflection matrix. Next, following the approach proposed in \cite{refR1}, the transformed vector $\mathbf{g} = \mathbf{A} \mathbf{h}_\mathrm{f}$ is distributed as $\mathcal{CN}(0, \mathbf{A}\mathbf{A}^H)$, and since $\mathbf{h}_\mathrm{b}$ and $\mathbf{h}_\mathrm{f}$ are independent, $G_\mathrm{b}$ becomes a Hermitian quadratic form in complex Gaussian vectors. Consequently, $G_\mathrm{b}$ follows a generalized chi-squared distribution, i.e., a weighted sum of independent exponential random variables.

To facilitate tractable analysis, $G_\mathrm{b}$ is approximated by a Gamma distribution using moment matching \cite{Ghadi2025:FRIS_Arxiv}. The first and second moments are computed as $\mathbb{E}[G_\mathrm{b}] = \operatorname{tr}(\mathbf{A} \mathbf{A}^H)$ and $\mathbb{E}[G_\mathrm{b}^2] = \operatorname{tr}^2(\mathbf{A} \mathbf{A}^H) + \operatorname{tr}((\mathbf{A} \mathbf{A}^H)^2)$, respectively. By leveraging the structure of $\mathbf{A}$, and as shown in \cite{refR1}, it holds that $\operatorname{tr}(\mathbf{A} \mathbf{A}^H) = \operatorname{tr}(\widetilde{\mathbf{J}}^2)$ and $\operatorname{tr}((\mathbf{A} \mathbf{A}^H)^2) = \operatorname{tr}(\widetilde{\mathbf{J}}^4)$, where $\widetilde{\mathbf{J}} = \mathbf{S}_{M_\mathrm{ON}} \mathbf{J} \mathbf{S}_{M_\mathrm{ON}}^T$ is the reduced spatial correlation matrix. Therefore, the Gamma shape and scale parameters are  given by $
k = \frac{\left(\operatorname{tr}(\widetilde{\mathbf{J}}^2)\right)^2}{\operatorname{tr}(\widetilde{\mathbf{J}}^4)}$ and $
\theta = \frac{\operatorname{tr}(\widetilde{\mathbf{J}}^4)}{\operatorname{tr}(\widetilde{\mathbf{J}}^2)}
$, and the proof is completed. 

\section{Proof of Proposition \ref{pro-2}}\label{app-2}
We now analyze the equivalent channel gain at the Eve as $
G_\mathrm{e} = \left|\mathbf{h}_\mathrm{e}^H \mathbf{A} \mathbf{h}_\mathrm{f} \right|^2,
$, where $\mathbf{h}_\mathrm{e} \sim \mathcal{CN}(\mathbf{0}, \mathbf{I}_M)$ denote the complex Gaussian channel vector from the FRIS to Eve. It should be mentioned that the matrix $\mathbf{A} \in \mathbb{C}^{M \times M}$ denotes the FRIS configuration (element selection and phase shifts), which is optimized solely based on Bob's CSI. Hence, $\mathbf{A} = \mathbf{J}^{1/2} \mathbf{S}_{M_\mathrm{ON}}^T \boldsymbol{\Phi} \mathbf{S}_{M_\mathrm{ON}} \mathbf{J}^{1/2}$, where $\boldsymbol{\Phi}$ and $\mathbf{S}_{M_\mathrm{ON}}$ are designed based on the instantaneous CSI of Bob and fixed during transmission. In contrast to Bob, Eve does not benefit from this optimization, and the BS is assumed to have no knowledge of Eve's channel. Consequently, although Eve observes the same physical configuration $\mathbf{A}$, from her perspective, this matrix is statistically random due to its dependence on the unknown CSI of Bob. Therefore, $\mathbf{A}$ is treated as a random matrix from Eve's point of view, with statistical behavior shaped by the ensemble of possible configurations optimized for Bob.

Now, to characterize the distribution of $G_\mathrm{e}$, we note that it can be written as a double-sided quadratic form:
\begin{align}
G_\mathrm{e} = \mathbf{h}_\mathrm{f}^H \mathbf{A}^{H} \mathbf{h}_\mathrm{e} \mathbf{h}_\mathrm{e}^H \mathbf{A} \mathbf{h}_\mathrm{f} = \mathbf{h}_\mathrm{f}^H \mathbf{B}_\mathrm{e} \mathbf{h}_\mathrm{f},
\end{align}
where $\mathbf{B}_\mathrm{e} \triangleq \mathbf{A}^{H} \mathbf{h}_\mathrm{e} \mathbf{h}_\mathrm{e}^H \mathbf{A}$ is a rank-one Hermitian positive semi-definite matrix. Conditioned on $\mathbf{h}_\mathrm{e}$, the equivalent gain $G_\mathrm{e}$ becomes a Hermitian quadratic form in a complex Gaussian vector and thus follows a generalized chi-squared distribution. However, since $\mathbf{A}$ is not deterministic from Eve's perspective, the matrix $\mathbf{B}_\mathrm{e}$ is itself a random quantity. Moreover, due to the independence between $\mathbf{h}_\mathrm{e}$ and $\mathbf{h}_\mathrm{f}$, and under the assumption that the optimization of $\mathbf{A}$ introduces effectively random alignment from Eve's viewpoint, the equivalent channel vector $\mathbf{g}_\mathrm{e} = \mathbf{A} \mathbf{h}_\mathrm{f}$ behaves statistically like a random projection of $\mathbf{h}_\mathrm{f}$. As a result, $\mathbf{g}_\mathrm{e}$ is a zero-mean complex Gaussian vector with a random covariance matrix. Thus, $G_\mathrm{e} = |\mathbf{h}_\mathrm{e}^H \mathbf{g}_\mathrm{e}|^2$ becomes a random quadratic form of a Gaussian vector with a random coefficient matrix, making exact analytical characterization intractable. However, for sufficiently large $M_\mathrm{ON}$ under  random projections toward Eve, the central limit theorem (CLT)-like behavior justify the use of an exponential approximation, i.e., $
G_\mathrm{e} \sim \text{Exp}(\theta_\mathrm{e})$, where  $\theta_\mathrm{e} = 1/\mathbb{E}[G_\mathrm{e}]
$. Th expectation of $G$ can be expressed as
\begin{align}
\mathbb{E}[G_\mathrm{e}] &= \mathbb{E}_{\mathbf{h}_\mathrm{f}, \mathbf{h}_\mathrm{e}} \left[ \left| \mathbf{h}_\mathrm{e}^H \mathbf{A} \mathbf{h}_\mathrm{f} \right|^2 \right] \\
&\overset{(a)}{=} \mathbb{E}_{\mathbf{h}_\mathrm{e}} \left[ \mathbf{h}_\mathrm{e}^H \mathbb{E}_{\mathbf{h}_\mathrm{f}} [\mathbf{A} \mathbf{h}_\mathrm{f} \mathbf{h}_\mathrm{f}^H \mathbf{A}^H] \mathbf{h}_\mathrm{e} \right], 
\end{align}
where $(a)$ is followed by the statistical independence of $\mathbf{h}_\mathrm{f}$ and $\mathbf{h}_\mathrm{e}$. Since $\mathbf{h}_\mathrm{f} \sim \mathcal{CN}(0, \mathbf{I})$, we have $\mathbb{E}_{\mathbf{h}_\mathrm{f}}[\mathbf{h}_\mathrm{f} \mathbf{h}_\mathrm{f}^H] = \mathbf{I}$ and thus $\mathbb{E}[G_\mathrm{e}]$ is derived as
$$
\mathbb{E}[G_\mathrm{e}] 
= \operatorname{tr} \left( \mathbb{E}_{\mathbf{h}_\mathrm{e}}[\mathbf{h}_\mathrm{e}\mathbf{h}_\mathrm{e}^H] \mathbf{A} \mathbf{A}^H \right)
= \operatorname{tr} \left( \mathbb{E}[\mathbf{A} \mathbf{A}^H] \right).
$$
In practice, since $\mathbf{A} = \mathbf{J}^{1/2} \mathbf{S}_{M_\mathrm{ON}}^T \boldsymbol{\Phi} \mathbf{S}_{M_\mathrm{ON}} \mathbf{J}^{1/2}$, we can rewrite $\mathbb{E}[G_\mathrm{e}]$ as
\begin{align}
\mathbb{E}[G_\mathrm{e}] &= \mathrm{tr} \left( \mathbf{J}^{1/2} \mathbf{S}_{M_\mathrm{ON}}^T \boldsymbol{\Phi} \mathbf{S}_{M_\mathrm{ON}} \mathbf{J} \mathbf{S}_{M_\mathrm{ON}}^T \boldsymbol{\Phi}^H \mathbf{S}_{M_\mathrm{ON}} \mathbf{J}^{1/2} \right)\\
&\overset{(b)}{=} \mathrm{tr} \left( \widetilde{\boldsymbol{\Phi}} \widetilde{\mathbf{J}} \widetilde{\boldsymbol{\Phi}}^H \widetilde{\mathbf{J}} \right)\overset{(c)}{=}\mathrm{tr}\left(\widetilde{\mathbf{J}}^2\right),\label{eq-exp}
\end{align}
where $\widetilde{\mathbf{J}} = \mathbf{S}_{M_\mathrm{ON}} \mathbf{J} \mathbf{S}_{M_\mathrm{ON}}^T$, $(b)$ follows from the cyclic trace and orthogonality of $\mathbf{S}_{M_\mathrm{ON}}$, while $(c)$ holds for a diagonal with unit-magnitude entries of matrix $\boldsymbol{\Phi}$. Therefore, by applying \eqref{eq-exp} into $G_\mathrm{e} \sim \text{Exp}(\theta_\mathrm{e})$, the proof is complete.

\section{Proof of Observation \ref{Pro_ASC}}\label{app_ASC}

We begin with the definition of the ASC as the expectation of the instantaneous secrecy capacity:
\begin{equation}
    \bar{C}_s = \mathbb{E}\left[ \left[ \log_2(1 + \gamma_\mathrm{b}) - \log_2(1 + \gamma_\mathrm{e}) \right]^+ \right].
\end{equation}
Applying Jensen's inequality to the concave function $\log_2(1 + x)$ yields an upper bound:
\begin{equation} \label{eq-upper}
    \bar{C}_s \leq \bar{C}_s^\mathrm{U} = \log_2\left(1 + \mathbb{E}[\gamma_\mathrm{b}] \right) - \log_2\left(1 + \mathbb{E}[\gamma_\mathrm{e}] \right).
\end{equation}
To evaluate this, we use the expectations of $\gamma_\mathrm{b}$ and $\gamma_\mathrm{e}$ based on their distributions. Since $\gamma_\mathrm{b} \sim \text{Gamma}(k_\mathrm{b}, \theta_\mathrm{b})$, we have:
$
    \mathbb{E}[\gamma_\mathrm{b}] = \bar{\gamma}_\mathrm{b} L_\mathrm{f} L_\mathrm{b} k_\mathrm{b} \theta_\mathrm{b}.
$
Likewise, since $\gamma_\mathrm{e} \sim \text{Exp}(\theta_\mathrm{e})$, we obtain
$
    \mathbb{E}[\gamma_\mathrm{e}] = \frac{\bar{\gamma}_\mathrm{e} L_\mathrm{f} L_\mathrm{e}}{\theta_\mathrm{e}}.
$
Substituting these into the \eqref{eq-upper} gives:
\begin{equation}
    \bar{C}_s^\mathrm{U} = \log_2 \left( 1 + \bar{\gamma}_\mathrm{b} L_\mathrm{f} L_\mathrm{b} k_\mathrm{b} \theta_\mathrm{b} \right) - \log_2 \left( 1 + \frac{\bar{\gamma}_\mathrm{e} L_\mathrm{f} L_\mathrm{e}}{\theta_\mathrm{e}} \right).
\end{equation}
Combining the logarithms and simplifying, we obtain the final expression as shown in \eqref{ASC pro}, which completes the proof.

\section{Proof of Observation \ref{Pro_SOP}}\label{app_SOP}

We start from the definition of the lower bound SOP, 
as
\begin{equation}
    P_\mathrm{sop}^{\mathrm{L}} = \int_{0}^{\infty} F_{\gamma_\mathrm{b}}\left(2^{R_s} \gamma_\mathrm{e} \right) f_{\gamma_\mathrm{e}}(\gamma_\mathrm{e}) \, \mathrm{d}\gamma_\mathrm{e}. \label{eq:SOP_L_def}
\end{equation}
Substituting the expressions for the CDF of $\gamma_\mathrm{b}$ and the PDF of $\gamma_\mathrm{e}$ yields
\begin{equation}
    P_\mathrm{sop}^{\mathrm{L}} \hspace{-3pt}=\hspace{-2pt} \frac{\theta_\mathrm{e}}{\bar{\gamma}_\mathrm{e} L_\mathrm{f} L_\mathrm{e} \Gamma(k_\mathrm{b})}
    \hspace{-3pt}\int_{0}^{\infty} \hspace{-4pt}\Upsilon \hspace{-3pt}\left( \hspace{-3pt}k_\mathrm{b}, \frac{2^{R_s} \gamma_\mathrm{e}}{\bar{\gamma}_\mathrm{b} L_\mathrm{f} L_\mathrm{b} \theta_\mathrm{b}} \hspace{-3pt} \right)
    \hspace{-2pt}\exp \hspace{-3pt}\left( \hspace{-4pt}- \frac{\theta_\mathrm{e} \gamma_\mathrm{e}}{\bar{\gamma}_\mathrm{e} L_\mathrm{f} L_\mathrm{e}}\hspace{-3pt} \right)
   \hspace{-2pt} \mathrm{d}\gamma_\mathrm{e}. \label{eq:SOP_L_substituted}
\end{equation}
Next, we express the lower incomplete gamma and exponential functions in terms of the Meijer G-function using \cite[Eq. 8.4.16.1]{ref59} and  \cite[Eq. 8.4.3.1]{ref59}, respectively. 
Substituting both Meijer G-forms into the integral and applying the identity from \cite[Eq. 2.24.1.1]{ref59} gives:
\begin{align}
    \mathcal{I}_1 &= \int_{0}^{\infty} 
    G^{1,1}_{1,2} \left( \frac{2^{R_s} \gamma_\mathrm{e}}{\bar{\gamma}_\mathrm{b} L_\mathrm{f} L_\mathrm{b} \theta_\mathrm{b}} \, \bigg| \, \begin{matrix} 1 \\ k_\mathrm{b}, 0 \end{matrix} \right)
    G^{1,0}_{0,1} \left( \frac{\theta_\mathrm{e} \gamma_\mathrm{e}}{\bar{\gamma}_\mathrm{e} L_\mathrm{f} L_\mathrm{e}} \, \bigg| \, \begin{matrix} - \\ 0 \end{matrix} \right) \mathrm{d}\gamma_\mathrm{e} \notag \\
    &= \frac{\bar{\gamma}_\mathrm{b} L_\mathrm{f} L_\mathrm{b} \theta_\mathrm{b}}{2^{R_s}} 
    G^{2,1}_{2,2} \left( \frac{\bar{\gamma}_\mathrm{b} L_\mathrm{f} L_\mathrm{b} \theta_\mathrm{b} \theta_\mathrm{e}}{\bar{\gamma}_\mathrm{e} L_\mathrm{f} L_\mathrm{e} 2^{R_s}} \, \bigg| \, \begin{matrix} -k_\mathrm{b}, 0 \\ 0, -1 \end{matrix} \right).
\end{align}
Finally, substituting $\mathcal{I}_1$ back into \eqref{eq:SOP_L_substituted}, we obtain the closed-form expression as \eqref{Pro-SOP}, which completes the proof.


\end{document}